\documentclass[11pt]{article}

\usepackage{amsmath,amsthm,amssymb,mathtools}
\usepackage{tcolorbox}
\usepackage{enumerate}
\usepackage{comment}
\usepackage{xspace}
\usepackage{subcaption}
\usepackage{algorithm}
\usepackage[noend]{algpseudocode}
\usepackage[disable]{todonotes}
\usepackage{microtype}
\usepackage[unicode]{hyperref}
\usepackage{xfrac}
\usepackage{main}
\def\capa{{\textsf{cap}}}
\def\tok{{\textsf{w}}}
\def\tot{{\textsf{numtokens}}}
\def\rwLength{{\textsf{rwLength}}}

\def\stepNum{{\textsf{stepNum}}}

\def\outbox{{\textsf{Outbox}}}
\def\core{{C_{t',t}}}
\def\polylog{\operatorname{polylog}}

\makeatletter
\newcommand\ackname{Acknowledgements}
\if@titlepage
  \newenvironment{acknowledgements}{%
      \titlepage
      \null\vfil
      \@beginparpenalty\@lowpenalty
      \begin{center}%
        \bfseries \ackname
        \@endparpenalty\@M
      \end{center}}%
     {\par\vfil\null\endtitlepage}
\else
  
\fi
\makeatother

\begin{document}

\date{}

\title{A Fully-Distributed  Protocol for Constructing Byzantine-Resilient  Peer-to-Peer Networks}
\author[1]{Aayush Gupta}
\author[2]{Gopal Pandurangan}

\affil[1,2]{\small Department of Computer Science, University of Houston, TX, USA\hspace{3cm}
\texttt{agupta56@cougarnet.uh.edu}, 
\texttt{gopal@cs.uh.edu}}

\maketitle

\begin{abstract}

We address a fundamental problem in  Peer-to-Peer (P2P) networks, namely, constructing and maintaining dynamic  P2P overlay network topologies with essential properties such as connectivity, low diameter, and high expansion, that are resilient to continuous high churn and the presence of a large number of malicious (Byzantine) nodes.  Our main goal is to construct and maintain a sparse (bounded degree) {\em expander} topology despite high churn and a large number of Byzantine nodes. Such an expander topology has logarithmic diameter, high expansion, and is robust to churn and the presence of a large number of bad nodes, and facilitates efficient and robust algorithms for fundamental problems in distributed computing, such as agreement,  broadcasting, routing, etc.
Existing protocols construct such expander networks that are tolerant to churn, but {\em do not} work under the presence of a large number of Byzantine nodes, which can behave arbitrarily and maliciously to disrupt the operation of the honest nodes.

We assume a stochastic churn adversary that models the network dynamism (nodes joining and leaving continuously) and a Byzantine full-information adversary (which controls the malicious nodes)  that has complete knowledge of the topology (including the identity of the nodes that join and leave) at any time and has unlimited computational power. It also has complete knowledge about the entire state of the network at every round, including random choices made by all the nodes up to and including the current round (but not future rounds).
  
Our main contribution is a randomized, fully-distributed dynamic P2P protocol that works with only local initial knowledge and guarantees, with a high probability, the maintenance of a {\em constant} degree graph with {\em high expansion} even under {\em continuous} churn and in presence of a large number of Byzantine nodes. Our protocol can tolerate up to $o(n/\polylog(n))$ Byzantine nodes  (where $n$ is the stable network size).  Our protocol is efficient, lightweight, and scalable, and it incurs only $O(\polylog(n))$ overhead for topology maintenance: only polylogarithmic  (in $n$) bits need to be processed and sent by each honest node per round, and any honest node's computation cost per round is also polylogarithmic. 

Our protocol can be used as a building block for solving fundamental distributed computing problems in highly dynamic networks, such as Byzantine agreement and Byzantine leader election, and enables fast and scalable algorithms for these problems.
\end{abstract}

\maketitle

\section{Introduction}
\label{sec:intro}

Modern  P2P networks, including those that maintain blockchains, are characterized by several features. First,  they are  {\em dynamic}:  not only do they exhibit high {\em churn}, i.e., nodes joining and leaving the network continuously, but also the network topology changes over time.  Second, they are vulnerable to attack: their {\em permissionless} nature allows malicious (\emph{Byzantine}) nodes to join at will and makes them open to attack. Third,  distributed computation of fundamental problems such as agreement, leader election, broadcasting, and routing must be accomplished efficiently despite high dynamism and the presence of adversarial participants in the network. Fourth, they are {\em sparse}, i.e.,  each node has only a small degree and typically has (initial) knowledge of only its local neighbors; this necessitates \emph{fully-distributed} protocols that work with only local knowledge.

Much of the well-established theory of distributed computing (see e.g., \cite{Lyn96,Pel00,AW04}) --- developed over the last five decades --- has focused on {\em static}, well-behaved networks. 
Dynamic networks pose non-trivial challenges when solving fundamental distributed computing problems such as: broadcasting, routing, agreement (consensus), and leader election \cite{Lyn96,Pel00,AW04,BA-survey,Wattenhofer_2019_Book}. 
Additionally, in dynamic networks where nodes and edges change continuously, even \emph{constructing and maintaining an efficiently functioning topology}
becomes a fundamental problem. In particular, we must design distributed algorithms to ensure that the dynamic network is well-connected and has low degree, low diameter, and high expansion  --- in short, a sparse {\em expander} graph. Accomplishing all the above becomes especially challenging in the presence of Byzantine nodes that may enter and leave the dynamic network at will.

Motivated by the above considerations, several works have taken steps towards designing provably efficient and robust protocols for highly dynamic (with high churn)
peer-to-peer networks.  The work of \cite{APRU12} studies the distributed agreement problem in dynamic P2P networks with churn. 
Its main contribution is an efficient and scalable  randomized  fully-distributed
protocol (i.e., each node processes and sends only polylogarithmic messages
per round, and local computation per node is also lightweight) that guarantees  stable almost-everywhere agreement\footnote{In sparse, bounded-degree networks,
an adversary can always isolate some number of non-faulty nodes, hence ``almost-everywhere" is the best one can hope for in such networks \cite{DPPU88}.}  with high
probability  even under very high {\em adversarial} churn rate (up to linear in $n$ {\em per round}, where $n$ is the network size) in a {\em polylogarithmic} number of
rounds. We note that this work does not handle Byzantine nodes. 
The work of \cite{podc13} presented an efficient fully-distributed protocol
for Byzantine agreement that works despite the presence of Byzantine nodes and high adversarial churn. This algorithm could tolerate up to $O(\sqrt{n}/\polylog(n))$ Byzantine nodes and up to $O(\sqrt{n}/\polylog(n))$ churn per round, took a polylogarithmic number of rounds,  and was scalable. The work of \cite{byzleader} presented an efficient distributed protocol for Byzantine leader election that could tolerate up to $O(n^{1/2 - \epsilon})$ Byzantine nodes (for a small constant $\epsilon > 0$)  and up to $O(\sqrt{n}/\polylog(n)$ churn per round, took a polylogarithmic number of rounds. The work of \cite{spaa13} focused on the problem of storing, maintaining, and searching data in P2P networks.  
It presented a storage and maintenance protocol that guaranteed with high probability, that data items can be efficiently stored (with only $O(1)$ copies of each data item) and maintained in a dynamic P2P network with churn rate up to ${O}(n/\polylog(n))$ {\em per round}. This protocol, however, cannot handle Byzantine nodes.

  A crucial ingredient that underlies all the above results (\cite{APRU12,podc13,byzleader,spaa13}) is the assumption that though the topology --- both nodes and edges ---  can change arbitrarily from round to round and is controlled by an adversary, the topology in {\em every round} is a (bounded-degree) {\em expander} graph. In other words, the adversary is allowed to control the churn as well as change the topology with the {\em crucial restriction} that it {\em always remains an expander graph in every round}. The assumption of an ever-present underlying expander facilitates the design of highly robust, efficient algorithms (tolerating a large amount of churn and Byzantine nodes). However, this is a very strong assumption, {\em that itself needs to be satisfied if one desires truly distributed protocols that work under little or no assumption on the underlying topology}. This motivates designing a distributed protocol that actually {\em builds} and  {\em maintains} an {\em expander topology under the presence of continuous churn and the presence of a large number of Byzantine nodes.}  Expanders have been used extensively to model dynamic P2P networks in which the expander property is preserved under insertions and deletions of nodes (e.g., see \cite{IPDPS14,LS03,PRU01} and the references therein). However, none of these works guarantee the maintenance of a {\em dynamic} expander network under the more challenging setting of a large number of Byzantine nodes. This is a fundamental ingredient that is needed to enable the applicability of previous results (\cite{APRU12, podc13, spaa13, byzleader}) under a more realistic setting.

We note that there are prior works (e.g., \cite{PRU01,LS03,focs15,Becchetti-ICDCS21}) that gave protocols for building and maintaining expander networks 
under continuous churn. However, these {\em do not} work under the presence of a large number of Byzantine nodes.

In this work, we seek to build and maintain a bounded degree {\em expander} P2P network in a dynamic {\em and} Byzantine setting, where nodes join and leave continuously and a large number of nodes can be Byzantine. 


\subsection{Model}\label{sec:model}

%





\noindent \textbf{Dynamic Network Model.}
Following \cite{PRU01} (also see \cite{Becchetti-ICDCS21,JP12}), we assume a stochastic churn model where the arrival of new nodes is modeled by Poisson distribution with rate $\lambda$, and the duration a node stays connected to a network is independently determined by an exponential distribution with mean $1/\mu$.  In this model, it can be shown (cf. Section 4.1 \cite{PRU01}) that the number of nodes in the network converges to $n = \lambda / \mu$, which we call the {\em stable network size}. In the stable state, the expected number of nodes that join {\em and} leave the network is $\lambda$ per unit time; thus $\lambda$ can be called the {\em churn rate} (in the long run).  This model is called the $M/M/\infty$, a well-studied model in queuing theory that has been used  to model
dynamic P2P networks \cite{PRU01,hari-podc,JP12,Becchetti-ICDCS21}.
 The model captures independent arrivals and (memoryless) departures. Furthermore, it can be generalized to arbitrary holding time distributions, e.g., Weibull or lognormal, which have been observed in practice \cite{JP12,SR06}. 

Throughout, without loss of generality, we assume an appropriate scaling of the time unit, so that
$\lambda = 1$, and hence $n = 1/\mu$ is the stable network size.  The unit of time can be suitably chosen in the model. In particular, we will assume that a unit of time is the time taken to send a message along an edge of the network, i.e., corresponding to one round in the synchronous communication model (discussed below).\footnote{In practice, typically, this time is very small, say in the order of milliseconds or less. Thus, the number of nodes entering or leaving (the churn rate) over a larger (more realistic) period of time (say, an hour) can be quite large.}

Let $G_t = (V_t, E_t)$ denote the network at time $t$, where $V_t$ represents the set of vertices and $E_t$ represents the set of edges in the network at time $t$. Initially, $G_0$ has zero vertices. Note that the P2P protocol determines the edge set $E_t$. 

Our goal is to design a fully-distributed protocol so that, for any time $t$, $G_t$ is a bounded degree network\footnote{Strictly speaking, we only care about the honest (non-Byzantine) nodes' degree being bounded.}  (degree bounded above by a fixed constant) with
good expansion with high probability.\footnote{We say that an event occurs {\em with high probability} (whp) if its probability is at least $1 - N^{-c}$, where $N$ is the network size, for a sufficiently large constant $c$.} The protocol will achieve this goal by dynamically adding/deleting edges over time.

Like prior works (e.g., \cite{PRU01,focs15,LS03}), we assume an entry mechanism that is needed for new nodes to connect to the network. When a new node enters, it contacts an {\em entry manager} which gives it the addresses of a few (a constant) number of nodes.\footnote{Such a service that provides an entry mechanism is common in real-world P2P networks. For example, in Bitcoin, the client software has  a list of ``seed'' nodes to connect to\cite{bitcoinbook}.} The entry manager chooses a small {\em random} subset of nodes that entered the network in the (not so distant) past. Unlike prior works\cite{Becchetti-ICDCS21,focs15}, the entry manager's role is quite minimal and used just for new incoming nodes to join the network. Crucially, the entry manager {\em does not} have any knowledge of  (or contact with) nodes once they enter the network; also, once a node enters the network, it does not communicate with the entry manager again.

\medskip

\noindent \textbf{Local Knowledge.} 
An important assumption in sparse networks is that nodes at the beginning have only \emph{local} knowledge, i.e., they have knowledge of only themselves and their neighbors in $G$. In particular, they do \emph{not} know the global topology or the identities of other nodes (except those of their neighbors) in the network. Our goal is to design {\em fully-distributed} protocols where nodes start with only local knowledge. However, as common in distributed computing literature, we assume that nodes
know an estimate of the stable network size $n$ (a constant factor upper bound of  $\log n$ will suffice).

\medskip

\noindent \textbf{Full Information Model and Byzantine Adversary.} We assume the powerful {\em full-information} model (e.g., see \cite{Ben-Or_2006,BL89,linial-full-info}) that has been studied extensively.  In this model,  the Byzantine nodes (controlled by an adversary) can behave arbitrarily and maliciously and have knowledge about the entire state of the network at every round, including random choices made by all the nodes up to and including the current round (this is also called {\em rushing} adversary), have unlimited computational power, and may collude among themselves (hence, cryptographic techniques are not applicable in this setting).  
We assume that the Byzantine adversary can choose to corrupt any node when it joins the network. Note that the adversary can see the whole topology and the new node's location in the network (including the neighbors of the new node) and decide whether to corrupt the node. If the adversary didn't choose to corrupt, then
the node remains uncorrupted (called {\em honest} or {\em good}) henceforth. Once a node is corrupted, it remains corrupted till it leaves the network.  Let $B_t$  be the set of
Byzantine nodes at time $t$, and we assume that $|B_t| = o(|V_t|/\log |V_t|)$, where $|V_t|$ is the network size (number of nodes) at time $t$. In particular,  once a stable network size is reached, $|B| = |B_n| =  o(n/\log n)$, where $n = \lambda / \mu$.
The total number of Byzantine nodes cannot exceed the above bound on $|B|$ at any time.

We note that the holding times for Byzantine nodes can be arbitrary (they need not follow the exponential distribution and can stay in the network as long as they want) as long as their total number is bounded (up to the above limit of $|B_t|$) at any time.  However, we assume that Byzantine nodes follow the Poisson arrival distribution like all nodes (which makes them harder to detect). 

\medskip





\medskip

\noindent \textbf{Communication Model.} Communication is {\em synchronous} and occurs via  message passing, i.e., communication proceeds in {\em discrete rounds}
by exchanging messages on {\em the edges} of the network, i.e., each node (including Byzantine nodes) can exchange messages
 only with its neighbors in the network.  In particular, a message sent by a node to its neighbor will be delivered
 in one round (which is our unit of time, as defined in stochastic churn model).
 
By our protocol design, honest nodes will only send $O(\polylog n)$ bits per edge per round. Note that Byzantine
nodes don't have any such limit and can send as many bits as they want. (Our protocol is designed in such a way
that crucially handles this extra power given to Byzantine nodes without limiting in any way the bandwidth capacity
of the edges.)
 As is standard in Byzantine algorithms (see, e.g., Lamport et al. \cite{PSL80}), we assume that the receiver of a message across an edge in $G$ knows the identity of the sender, i.e., if $u$ sends a message to $v$ across edge $(u,v)$, then $v$ knows the identity of $u$; also the message sent across an edge is delivered correctly and in order. 

We assume that the communication links are \emph{reconfigurable}: if a node $u$ knows about the ID of some node $v$, then $u$ can establish or drop a link to $v$.\footnote{Strictly speaking, it takes a successful handshake between $u$ and $v$ to establish or drop a bidirectional link.
For simplicity, and since it does not change the asymptotic bounds of our results, we assume that these connections happen instantaneously.} Also, as is standard in  P2P (and overlay) networks, a node can establish a connection (directly) with another node if it knows the ID (e.g., IP address)  of the other node.

\subsection{Our Contributions} \label{sec:result}

Our main contribution is a randomized, {\em fully-distributed} P2P construction protocol that, with high probability, builds and maintains a {\em constant} degree graph with {\em high expansion} even under a {\em continuous} large stochastic churn  and the presence of a large number of Byzantine nodes which operate under the powerful full-information model. Our protocol can tolerate up to $o(n/\polylog(n))$ Byzantine nodes  (where $n$ is the stable network size).  Our protocol is efficient, lightweight, and scalable, and it incurs only $O(\polylog(n))$ overhead for topology maintenance: only polylogarithmic  (in $n$) bits need to be processed and sent by each honest node per round, and any honest node's computation cost per round is also polylogarithmic. Our main result is stated  in Theorem 
\ref{thm:main}.

Expander graphs are crucial in tolerating a high amount of churn
and a large number of Byzantine nodes. As discussed in Section \ref{sec:intro}, several prior works gave efficient fully-distributed algorithms for various problems in P2P networks such as Byzantine agreement\cite{podc13}, leader election \cite{byzleader}, storage and search \cite{spaa13}, and efficient routing and Distributed Hash Table (DHT) construction\cite{Augustine-SPAA22} under the {\em crucial assumption of an ever-present 
underlying expander topology} despite churn and the presence of a large number of Byzantine nodes.
Our P2P construction protocol actually constructs and maintains an expander in a dynamic setting under the presence of a large number of Byzantine nodes. Hence, it enables the applicability of the prior P2P protocols
{\em without the assumption} of an ever-present expander topology.\footnote{However, while these prior results work even under adversarial churn, our protocol works under a stochastic churn model.}

Indeed, our protocol can be used as a building block for solving fundamental distributed computing problems such as Byzantine agreement, Byzantine leader election, and Byzantine search and storage in dynamic P2P networks, and enables fast and scalable algorithms for these problems. In particular,  Byzantine agreement and leader election can be efficiently solved in dynamic P2P networks by a straightforward adaptation of prior protocols \cite{podc13, byzleader} to run on top of our P2P construction protocol that maintains an expander. This leads to  Byzantine agreement and leader election algorithms with the same respective guarantees as in \cite{podc13} and \cite{byzleader}. Specifically, we obtain fast Byzantine agreement and leader algorithms running in $\polylog{n}$ rounds in the stochastic churn model that tolerates $O(\sqrt{n}/\polylog{n})$ Byzantine nodes at any time.\footnote{Note that this bound on the number of Byzantine nodes is well within
the Byzantine tolerance of $o(n/\polylog{n})$ nodes of the P2P construction protocol.}


\subsection{Challenges and High-Level Overview of our Approach}

In this work, we solve the problem of constructing and maintaining a sparse, bounded-degree expander in a dynamic P2P network, despite continuous churn and Byzantine nodes, with local knowledge using message passing in the full information model. Our protocol works by ensuring two interdependent invariants (with high probability). The first invariant is that the nodes can sample from the set of IDs in the network almost uniformly at random. We implement this sampling via {\em random walks}. For the sampling technique to work efficiently, we need our second invariant, namely, that the network maintains good expansion. For this, we employ a technique of connecting each node to a constant number of other nodes chosen (almost) uniformly at random, thus bringing us back to our first invariant. While this high level idea is quite straightforward, there are some significant challenges to be overcome to get the protocol to work. The first main challenge is to sample (nearly) uniformly at random via random walks. Byzantine nodes
can bias these random walks. We make use of a crucial protocol called the Byzantine Random Walk Protocol (explained below)
that allows honest nodes to sample nodes (almost uniformly at random). The second challenge is to maintain the expander despite continuous churn. Even if we construct an expander at some time, if it is not maintained, the expansion will degrade due to node deletions and additions. \\


\noindent \textbf{Byzantine Random Walk Protocol.} Our P2P-construction protocol relies on maintaining connections to a constant number of randomly sampled nodes in the network at any time. This is challenging due to a large number of Byzantine nodes and continuous joining and leaving of nodes. In particular, we need honest nodes to sample honest nodes (almost) uniformly at random (in the current network) despite the presence of Byzantine nodes and churn.  This is accomplished by doing {\em random walks}. 
We give
an intuition as to why random walks work well in a sparse network (unlike broadcast, for example).  Random walks
are lightweight (and local) and allow us to bound the number of messages sent by {\em Byzantine nodes}. Byzantine nodes need not follow
the random walk protocol and can send a lot of messages, but once these messages reach honest nodes, their influence
becomes limited. 

We address this issue by using the Byzantine Random Walk Protocol from \cite{soda25} and adapting it. The main idea is that by executing random walks within the network in a controlled way (so as to limit the influence of Byzantine nodes), most honest nodes can sample most (other) honest nodes nearly uniformly at random. (To control the number of tokens sent by Byzantine nodes --- which may not follow the random walk protocol --- honest nodes blacklist neighbors that send a number of tokens above a particular threshold per round per edge.) The protocol of \cite{soda25} is limited to static networks. We extend our protocol in Algorithm \ref{alg:byzantineSamplingSparse} to a dynamic network and by verifying if a random walk ended successfully using the \textit{verifiedList}. Each node initiates multiple independent random walks, that traverse the network for $O(\log{n})$ steps. When a random walk token initiated at a node $v$ ends at some node $u$, then node $v$ adds $u$ to its \textit{verifiedList}.  Then $v$ returns the ``verified''
token back to $u$ which $u$ can then use to connect to $v$. This verification ensures that only nodes whose tokens reach $v$
via random walks are allowed to connect to $v$; this crucially prevents Byzantine nodes (who have full information on the addresses of all nodes in the networks) from flooding honest nodes with connection requests. \\

\noindent \textbf{P2P Construction Protocol.} We present our fully-decentralized P2P-construction protocol in Algorithm \ref{alg:p2pConstr}. The main idea of the protocol is to use the randomly sampled nodes to make connections; we show 
that this creates a large expander subgraph consisting of honest nodes. There are two challenges in making these connections. The first is that all connection requests cannot be satisfied; otherwise, the degree will not be bounded by a constant. Hence we need to reject some connections. We show that despite such rejections, an expander subgraph is created. We accomplish this by allowing somewhat more incoming connections (``incoming degree'') than outgoing connections (``outgoing degree"). Second, the protocol has to maintain the expansion despite churn. Churn leads to degradation of the expansion properties mainly due to node deletions. To overcome this, the main idea of the protocol
is to {\em periodically} drop connections and make new connections. This is done every $\Theta(\log n)$ rounds --- called
a {\em phase}. The random walk tokens generated by (most) honest nodes at the beginning of the phase finish all their walks and return back (with verified tokens) to the respective source nodes at the end of the phase. These verified tokens are then used to make (nearly random) connections at the end of the phase (these connections are direct P2P connections and hence, we can assume, takes place instantly). We discard the verified tokens at the end of every phase; this helps in bounding the dependencies that arise over time.

It is easy to see that our P2P construction protocol (including the Byzantine Random Walk Protocol) is fully-distributed (i.e., operates with only local topological knowledge), very lightweight,
and message-efficient. In each round, each node processes only polylogarithmic  (in $n$) bits (where $n$ is the stable network size) and any honest node's computation cost per round is also polylogarithmic. It also communicates only 
polylogarithmic bits per round to its constant number of neighbors.

We note that our random walk based protocol cannot tolerate $\Omega(n/\log n)$ Byzantine nodes, since the mixing time needed in a sparse expander network is at least $\Theta(\log n)$.  
If there is a linear number of Byzantine nodes, then most random walks will go through a Byzantine node, and the Byzantine Random Walk Theorem does not work. We conjecture that no {\em fully-distributed} algorithm may tolerate $\omega(n/\log n)$ Byzantine nodes in a sparse network.
Our protocol reaches close to this limit, i.e., it can tolerate up to $o(n/\log n)$ Byzantine nodes.

\subsection{Additional Related Work} \label{sec:related}

As mentioned earlier, several prior works (e.g., \cite{PRU01, JP12, LS03,focs15,Becchetti-ICDCS21}) that gave protocols for building and maintaining expander networks 
under continuous churn. However, these {\em do not} work under the presence of a large number of Byzantine nodes.
In particular,  the early work of \cite{PRU01} designed a P2P protocol that constructed and maintained expander networks under a stochastic churn adversary that modeled a continuous high churn. The stochastic churn model (the $M/M/\infty$ model)  introduced in this paper has been used in several other papers \cite{hari-podc,Becchetti-ICDCS21,JP12} including our paper.  The more recent work of \cite{Becchetti-ICDCS21} under the same stochastic churn model shows how expansion can be maintained. However, this work makes a simplifying assumption that not only does a joining node get access to a random node in the network,
but also, that any node that loses a neighbor gets access to a {\em random node}  in the current network. This assumption helps in showing that expansion is maintained. However, this protocol allows some nodes to have $\Theta(\log n)$ degree and hence not bounded (Unlike ours).
More importantly, this work {\em does not} handle Byzantine nodes. In contrast, a main novelty in our protocol is that nodes have to themselves decide which nodes to connect to if they lose a neighbor after joining the network. Furthermore, a crucial feature of our work is handling a large number of Byzantine nodes.

  
  There has been a lot of work on P2P protocols for maintaining desirable properties (such as connectivity, low diameter, high expansion, bounded degree) under churn (see e.g., \cite{PRU01, JP12, KSW10} and the references therein), but these do not work under the presence of a large number of Byzantine nodes. 
Most prior algorithms (e.g., \cite{LS03, skipplus, PT11, IPDPS14, hyperring,mahlmann}) will only work under the assumption that the network will eventually stabilize and stop changing or there is a ``repair" time for maintenance when there are no further changes (till the repair/maintenance is finished); these algorithms do not work under continuous churn or, more importantly, under the presence of Byzantine nodes. 

There have also been works on overlay network construction problem, where the goal is to transform any arbitrary graph into an expander by adding and deleting edges (i.e., the graph is reconfigurable) \cite{angluin,constructor,gotte,itcs24}. However, these works do not assume churn or Byzantine nodes.

There has also been significant prior work in designing P2P networks that
are provably robust to a large number of Byzantine faults (e.g., see ~\cite{FS02,HK03,NW03,Scheideler05,AS09}). These focus on robustly enabling storage and retrieval of data items under adversarial nodes. However, these algorithms will not work in a highly dynamic setting with large, continuous churn.   
The work of \cite{Guerraoui-PODC13} presents a solution for maintaining a clustering of the network where each cluster contains more than two thirds honest nodes with high probability in a setting where the size of the network can vary polynomially over time; however the churn is limited to $O(\polylog n)$ per round and the network degree is also polylogarithmic in $n$.
The work of \cite{KSSV06}, which achieves distributed Byzantine agreement and leader election in \emph{static} P2P networks, raised the open question of whether one can design robust  P2P  protocols that can work in highly dynamic networks with a large churn. 

The work of \cite{KSW10} shows that up to $O(\log n)$ nodes (adversarially chosen) can crash or join per constant number of time steps. Under this amount of churn, it is shown in \cite{KSW10} how to maintain a low peer degree and bounded network diameter in P2P systems by using the hypercube and pancake topologies.  

 In \cite{SS09} it is shown how to maintain a distributed heap that allows join and leave operations and, in addition, is resistant to Sybil attacks. A robust distributed implementation of a distributed hash table (DHT) in a P2P network is given by \cite{AS09}, which can withstand two important kinds of attacks: adaptive join-leave attacks and adaptive insert/lookup attacks by up to $\epsilon n$ adversarial peers. This paper assumes that the good nodes always stay in the system and the adversarial nodes are churned out and in, but the  {\em algorithm} determines where to insert the new nodes. 
The work of \cite{NW03} describes a simple DHT scheme that is robust under the following simple random deletion model --- each node can fail independently with probability $p$. They show that their scheme can guarantee logarithmic degree, search time, and message complexity if $p$ is sufficiently small.
\cite{HK03} describe how to modify two popular DHTs, Pastry \cite{Pastry} and Tapestry \cite{ZKJ01} to tolerate random deletions.
 Several DHT schemes (e.g., \cite{SM+01,RF+01,Koorde}) have been shown to be robust under the simple random deletion model mentioned above.
 There have also been works on designing fault-tolerant storage systems in a dynamic setting
 using quorums (e.g., see \cite{dynamic-quorum03, NU-quorum05}). However, these do not apply to our model of continuous churn.

The work of  \cite{Augustine-SPAA22} gives a protocol for constructing a Distributed Hash Table (DHT) given an arbitrary {\em static} expander network (which is not addressable) in the presence of a large number (up to $n/\polylog{n}$) of Byzantine nodes. The model used in this paper assumes a reconfigurable (i.e., a P2P) network (where a node can add or drop edges to other nodes whose identifier it knows).  This paper assumes {\em private channels}, which is significantly weaker than the full information model since it assumes that communications between honest nodes are unknown to Byzantine nodes. This paper does not address how the DHT  can be maintained under churn, where nodes can join or leave continuously.

Finally, the recent work of \cite{soda25} gives a fully-distributed protocol for solving
Byzantine agreement problem in {\em static expander graphs}. However, it does not handle churn. This paper crucially uses random walks to circumvent the behavior of Byzantine nodes. This work shows how to implement random walks in bounded-degree networks with
a large number of Byzantine nodes and presents the ``Byzantine Random Walk Theorem''. This theorem shows precisely how the Byzantine Random Walk Protocol controls the messages sent by Byzantine nodes
and how their influence can be limited for most random walks initiated by most honest nodes. We also use this result in our paper. We note that a crucial difference between their work and ours is that they assume an expander graph that does not change, whereas we have to construct and maintain an expander under continuous churn and Byzantine nodes.

\subsection{Preliminaries}



\noindent {\textbf Dynamic Random Graph.} A \textit{continuous dynamic graph} $\mathcal{G}$ is a family of graphs $\mathcal{G} = \{ G_t = (V_t, E_t) : t \in \mathbb{R}^+\}$, where $V_t$ is the set of nodes and $E_t$ is the set of edges in the system at time $t$.  If $\{ V_t\}_t$ and $\{ E_t\}_t$ are families of random sets we call the corresponding random process a \textit{dynamic random graph}. We call $G_t$ the \textit{snapshot} of the dynamic random graph at time $t$. For a set of nodes $S \subseteq V_t$, we denote with $\delta_{out}^t(S)$ the outer boundary of $S$ in snapshot $G_t$.

Let $G = (V,E)$ be an undirected graph. For any node $v \in V$, its degree is defined as $d(v) = |N(v)|$, where $N(v)$ is the set of neighbors of $v$. For a subset $ S \subseteq V$, the volume of $S$ is defined as $vol(S) = \sum_{v \in S}d(v)$ and the edge cut of $S$ as $E(S, V\backslash S) = \{ (u,v) | (u,v) \in E, u \in S, v \notin S\}$. The conductance of (a non-empty set) $S$, $\phi(S)$, is defined as $\phi(s) = |E(S, V \backslash S)| / \min(vol(S), vol(V \backslash S))$. Finally, the conductance of the entire graph $G$ is $\phi = \min_{S \subseteq V, |S| \leq |V|/2} \phi(S)$. As we define next, graphs with {\em constant} conductance are called {\em expander } graphs.

\begin{definition}[Expander Graph]\label{def:expander-graph}
    A family of graphs $G_n$ on $n$ nodes is an \emph{expander family} if,
    for some constant $\alpha$ with $0 < \alpha < 1$, the conductance
    $\phi_n = \phi(G_n)$ satisfies $\phi_n \geq \alpha$ for all
    $n \geq n_0$ for some $n_0 \in \mathbb{N}$.
\end{definition}

\begin{definition}[Stochastic Churn Model \cite{PRU01}]\label{def:poisson_node_churn}
    A Poisson node churn is a random process ${V_t : t \in \mathbb{R}}$ such that: 
    \begin{enumerate}
        \item $V_0 = \phi$;
        \item The arrival of new nodes in $V_t$ is a Poisson process with rate $\lambda$.
        \item Once a node is in $V_t$, its lifetime has exponential distribution with parameter $\mu$. \cite{Becchetti-ICDCS21}
        \end{enumerate} 
\end{definition}

In Definition \ref{def:poisson_node_churn}, the arrival time interval between two consecutive nodes is an exponential random variable of parameter $\lambda$, and the number of nodes joining in a time interval of duration $\tau$ is a Poisson random variable with expectation $ \tau \cdot \lambda$. Further, $\{ V_t : t \in \mathbb{R}^+\}$ is a continuous Markov Process. 

We show that in the long run the network size converges to $\lambda / \mu$ which we denote by $n$, called the {\em stable network size}.

\begin{lemma}[stable network size \cite{PRU01}]\label{lem:sizeofg}
    Consider  the Poisson node churn $\{V_t : t \in \mathbb{R}^+\}$ in Definition \ref{def:poisson_node_churn} with parameters $\lambda$
    and $\mu$. 
    \begin{enumerate}
        \item For any $\log n \leq  t \leq n$, w.h.p, $|V_t| = \Theta(t)$.
        \item For any $t > n$, w.h.p., $|V_t| = \Theta(n)$.
    \end{enumerate}
\end{lemma}


Without loss of generality (cf. Section \ref{sec:model}), we will assume $\lambda = 1$, i.e., the mean number of nodes that arrive per {\em unit of time} is $1$, and hence $n = 1/\mu$.

\smallskip
We next show a bound on the number of nodes that arrive within an interval of time (proof in Section \ref{sec:arrivalproof}).

\begin{lemma}
\label{lem:arrival}
    Let $N(t', t)$ be the number of nodes arriving in the interval $[t', t]$ in the Poisson churn model with rate $\lambda = 1$
    and stable network size $n$. Then, $E[N(t',t)] = t-t' = n'$ and  $N(t',t)$ is well concentrated around $n'$:
        $$Pr\left(n' - O(\sqrt{n'\ln n}) \leq N(t,t') \leq n' + O(\sqrt{n'\ln n})\right) \geq 1-  \frac{1}{n^2}.$$
\end{lemma}

We note that in our protocol and analysis, we consider the dynamic graph $G_t$
at discrete integer times --- $G_0, G_1, G_2, \dots$ --- these are just snapshots of
the dynamic graph at these integer time steps.

We next show some properties of the lifetime (or holding time)  of a node that follows
from the property of the exponential distribution.

\begin{lemma}
\label{lem:expo}
Let the holding time for a node (i.e., the time the node is alive in the network)
be $H$. Then $E[H] = n$ and $\Pr(H > t) = e^{-t/n}$, where $n$ is the stable network size.
In particular, $\Pr(H > n \ln n) = 1/n$. Also, $\Pr(H \leq t) = 1 - e^{-t/n}$. Hence,
$\Pr(H \leq n) = 1-1/e$. 
\end{lemma}

\subsection{Expander Graphs}
\label{subsec:defCore}
Consider a network graph $G = (V,E)$ with $|V|=n$ and $|E|=m$, such that at most $|B| = o( n/\log n)$ nodes are Byzantine. 
Let $G$ be a constant-degree expander graph (degree bounded by a constant) with constant conductance $\phi_G$ and mixing time $\tau = O(\log n)$. 
 The mixing time of $G$ is defined as $\tau = \arg \min_t (||A^t \pi - \mathbf{u}||_\infty \le 1/n^3)$, where $A$ is the adjacency matrix of $G$, $\pi$ is any arbitrary probability distribution over $V$, and $\mathbf{u}$ is the stationary distribution over $V$. 
The stationary probability of a node $u$ would be $deg_G(u)/2m$, where $deg_G(u)$ is the degree of $u$ in $G$.  Since all nodes have constant degree, the stationary probability will be nearly uniform, i.e., $\Theta(1/n)$.

Now, if one considers only the honest nodes of $G$ then it is known that a subset of them induces an expander subgraph. More concretely, for a constant-degree expander $G = (V,E)$  and at most $|B| = o(n)$ Byzantine nodes, it follows from Lemma 3 in \cite{focs15}, that for {\em any}  constant $c < 1$, there exists a subgraph $C$ in $G \setminus B$ that is of size  $n - O(|B|)$ and that has constant conductance $\phi_C = c \cdot \phi_G$.  This expander subgraph of $G$ is called the \emph{core} of $G$, and denoted by $C=(V_C,E_C)$. 
Note that the core $C$ consists of only good nodes (but note that good nodes themselves do not know if they belong to core or not). Moreover, since $C$ is an expander, a random walk restricted to $C$ will have mixing time $\tau_C = b\log n$ for some suitably large constant $b$ (depending on $\phi_C$).
Note that $\tau_C = \arg \min_t (||A_C^t \pi - \mathbf{u}||_\infty \le 1/n^3)$, where $A_C$ is the adjacency matrix of core $C$, $\pi$ is any arbitrary probability distribution over $V_C$, and $\mathbf{u}$ is the stationary distribution over $C$. To distinguish $\tau$ and $\tau_C$, we refer to $\tau_C$ as the \emph{core mixing time}.

\section{P2P Construction Protocol }

\begin{algorithm}
\caption{P2P Construction Protocol}
\label{alg:p2pConstr}        

\vspace{10pt}

\textbf{Node $u$ enters the network:}
\begin{algorithmic}[1]
    \State $candidateList \leftarrow $  Query EntryManager for $3d$ candidate nodes \Comment{\textit{ $d$ is a sufficiently large constant}}
    \State Contact nodes in $candidateList$ to establish connection
    \State $d_{out}\leftarrow$ Number of successful connections
    \While{$d_{out} \leq d $} 
        \State $candidateList \leftarrow $ Query EntryManager for $3d$ candidate nodes
        \State Contact nodes in $candidateList$ to establish connection
        \State $d_{out}\leftarrow$ Number of successful connections
    \EndWhile
    \State Initialize \texttt{Byzantine Random Walk Protocol}
    \Comment{\textit{At least one connection has been established}}
\end{algorithmic}

\vspace{10pt}

\textbf{After node $u$ connects to the network}:
\begin{algorithmic}[1]
\For{$t = \eta \log n, 2\eta \log n, 3 \eta \log n \dots $} \Comment{\textit{for every $\eta \log n$ steps --- called a phase}} 
    \If{$d_{out}(u) \geq 2d$}
        \State Drop $d$ randomly selected existing connections 
        \State Establish $d$ new connections from the \texttt{verifiedList}
        \Statex \Comment{\textit{Using verified tokens from Byzantine Sampling Protocol from the previous phase} }
    \EndIf
    \If{$ 0 \leq d_{out}(u) < 2d$}
        \State Establish $3d - d_{out(u)}$ new connections from the \texttt{verifiedList}
    \EndIf
\If{$0 \leq d_{in}(u) \leq 6d$} 
    \State Accept random $6d - d_{in}(u)$ connection requests if these connections are from \texttt{verifiedList} \Statex \Comment{\textit{ $u$ only  accepted connections from nodes  that it verified earlier via random walks}}
\EndIf

\If{$\text{connection request from {\em new node}}$ \textbf{and} $d_{in}(u) < 6d$}
        \State Accept the incoming request \Comment{\textit{Exception for new node requests}}
    \EndIf

\If{number of incoming requests from $v \geq 6d$ }
    \State Reject all incoming connection requests from $v$
\EndIf

\State Empty the \texttt{verifiedList} \Comment{\textit{Refresh the connection tokens after every phase}}

\EndFor
\end{algorithmic}
\end{algorithm}

\begin{algorithm}
    \caption{Entry Manager Protocol}
    \label{alg:entrymgr}
    \textbf{Input:} New node $u$ joins the network.\\
    \textbf{Initialize:} $nodesList \leftarrow \texttt{NULL}$

    \textbf{Node $u$ enter the network:}
    \begin{algorithmic}[1]
        \If {size of $nodesList$ is  equal to $n$}
            \State Remove a random node from $nodesList$
        \EndIf
        \State Add node $u$ to $nodesList$
    \end{algorithmic}
        
    \textbf{Connection query from node $u$:}
    \begin{algorithmic}[1]
        \State $candidateList \leftarrow$ Randomly select $3d$ nodes from $nodesList$ 
        \State Return $candidateList$
    \end{algorithmic}
\end{algorithm}


The {\em P2P Construction Protocol} (Algorithm \ref{alg:p2pConstr}) enables nodes to join and manage connections within a peer-to-peer network while defending against Byzantine nodes. The protocol tries to build a constant degree network with good expansion. The total degree of each node is at most $9d$, where $d$ is a suitably large constant (fixed in the analysis).\footnote{Although we don't optimize
the value of $d$, $d$ need not be a large constant (e.g., less than 10). Notice that Bitcoin allows 8 outgoing connections.} The total degree of a node consists of at most $3d$ \emph{outgoing} connections and at most $6d$ \emph{incoming} connections (in other words, a node can accept up to $6d$ connections from other nodes). We note that the categorization of connections as outgoing or incoming is for protocol purposes; all connections are treated as undirected edges for communication and topology analysis.

The \emph{Entry Manager Protocol} (Algorithm \ref{alg:entrymgr}) provides a new node, $u$, with a set of $3d$ random nodes. Initially, it is an empty list, named \emph{nodesList}, and the protocol first checks if the size of this list is  equal to $n$. If the condition is true, a random node from the \emph{nodesList} is removed and then $u$ is added to the list. Otherwise, if the list is smaller than $n$, $u$ is added to the list.  Further, it provides $u$ with randomly selected $3d$ nodes from the \emph{nodesList}.   

When a new node $u$ joins the network, it is provided $3d$ nodes ({\em candidateList}) by the $EntryManager$ to establish initial (outgoing) connections, using the $EntryManager$ (Algorithm \ref{alg:entrymgr}). When a new node initiates a connection with those $3d$ nodes, the ones which are existing (i.e., those that are still alive in the network) will accept the connection request if they have less than $6d$ incoming connections.  (We note that the $EntryManager$ does not know anything about the nodes in its \emph{nodesList}, including whether they are still in the network or have left.) If $u$ fails to connect with any of these nodes, it retries by re-querying the $EntryManager$ for another set of $3d$ nodes until it successfully establishes at least $d$ connections to existing nodes in the network. 

After a (honest) node $u$ successfully connects to the network, it executes Byzantine Random Walk Protocol (described in Section \ref{sec:byzwalk})
that allows it to sample (honest) nodes almost uniformly from the network. The random walk protocol supplies ``verified'' tokens (cf. Section \ref{sec:byzwalk}) to nodes, which allows nodes to make new connections. 


A main feature of this protocol that helps maintain high expansion despite the dynamic changes and presence of Byzantine nodes is that every node repeatedly replenishes its connections: it drops connections and makes new connections. If a node's outgoing degree exceeds $2d$, it drops $d$ randomly selected connections and attempts to replace them with $d$ new connections.  If the outgoing degree is between 0 and $2d$, the node attempts to establish additional connections until it reaches $3d$ outgoing connections. 

Further, if $u$'s incoming connection has less than $6d$ connections, it randomly accepts up to $6d - d_{in}(u)$ requests from other nodes. Also, if the number of incoming connection requests exceeds $6d$, it rejects all the connections. When a new node tries to contact $u$, and it has an incoming degree less than $6d$, it has to allow it to establish the connection, if the connection is a verified one and if the incoming degree is $6d$; otherwise it will reject. Note that verified connections, which are ratified by the nodes (via random walks) accepting them, prevent Byzantine nodes from flooding honest nodes with connection requests.



\subsection{Byzantine Random Walk Protocol}
\label{sec:byzwalk}

We present a distributed protocol to do random walks in a sparse network under the presence of a large number of Byzantine nodes. The protocol is an adaption of a similar protocol presented in \cite{soda25} that applied
only to static networks. Our Byzantine Random Walk Protocol is presented in Algorithm~\ref{alg:byzantineSamplingSparse}. The protocol allows an almost uniform random sampling of {\em honest} nodes for {\em most honest nodes} by sending tokens via random walks. More precisely, each (honest) node in the network, in each discrete round ($t = 1, 2, \dots$), initiates a number of independent random walks, up to a maximum of $\tot$ random walks (tokens) per node.  Each such token
walks for $2f = O(\log n)$ rounds which we call a phase. 
Each of these tokens perform independent random walks on the dynamic graph ${\mathcal G}$. 
A random walk in a dynamic graph is similar to that of a static graph and 
 is defined as follows:  assume that a time $t$ the token is at node $v \in V$, and let $N(v)$ be the set of neighbors of $v$ in $G_t$, then the token goes to one of its neighbors from $N(v)$ uniformly at random. 
We point out that although the random walks happen in the dynamic graph, since each random walk is short (i.e., takes
only $O(\log n)$ steps), we will analyze the random walk as though it is happening in a {\em static} subgraph called
the {\em core} (cf. Section \ref{subsec:defCore}).
  
To handle Byzantine nodes, each honest node locally regulates the rate at which the tokens flow in and out of it. Specifically,  at most $\capa$ tokens are allowed to enter/exit the node through each of its incident edges per round. We employ a {\em FIFO buffer} at each incident edge to hold tokens that could not be sent in the current round. As a result, a token may be held back at multiple buffers during the phase. Nevertheless, Theorem \ref{thm:ByzSamplingFinal} shows that all the random walks that {\em only walk on the core subgraph} (called ``good'' random walks) will make \emph{at least $f$ random steps} (or in other words, can only be held back during $f$ rounds) with high probability. 

Then, it follows that if we choose $f$ to be the mixing time of the core subgraph $C$ then this will ensure the mixing of those walks in $C$. Most random walks initiated by nodes in $C$ will walk only in $C$. This implies our Byzantine Random Walk Theorem (see Theorem \ref{thm:ByzSamplingFinal}), which says that most random walks initiated in the core $C$ walk only in $C$ and mix rapidly, at which point they reach the stationary distribution over $C$.

\begin{algorithm}
\caption{Byzantine Random Walk Protocol for \textbf{node} $u$ with TTL }
\label{alg:byzantineSamplingSparse}
\begin{algorithmic}[1]

\Require 
\Statex $\tot = \log^3 n$ \Comment{\textit{Number of tokens to be initiated at $v$ per round.}}
\Statex $\capa =  a \log^{3} n$ 
\Statex \Comment{\textit{Number of tokens allowed through an edge in one round (for a  large enough fixed constant $a > 0$).}}
\Statex $\rwLength = c\log n$ 
\Statex \Comment{\textit{Length of random walk for each token for some large enough constant $c$.}}
\Statex $\textsc{rwCounter}$ for each token \Comment{\textit{To measure the distance of random walk}}
\Statex $\outbox_v$ for each neighbor $v$ \Comment{\textit{FIFO token buffers stored at $u$, one for each neighbor $v$.}}

\Statex
\For{$t = 1, 2, \dots$, for each phase:} \Comment{\textit{After a node joins the network}}

\State $u$ creates $\tot$ tokens, each with value of $\textsc{rwCounter} = 0$.
\Statex \Comment{\textit{Each token has distance ID and stores the source node ID}}

\ForAll{tokens $\tok$ that were created}
	\State Pick a neighbor $v$ uniformly and independently at random.
 	\State Push  $\tok$ into $\outbox_v$.
\EndFor
\For{$\stepNum \gets 1$ to $\rwLength$}
\For{each neighbor $v$}
\State Dequeue up to $\capa$ tokens from $\outbox_v$ (which is a FIFO queue).
\State Record $v$ as the next node in the walk taken by each of those tokens.

\EndFor
\State Receive up to $\capa$ tokens sent by each neighbor and store them in a set $M$. 
\Statex \Comment{\textit{Any neighbor that sends more than $\capa$ node is blacklisted and heretofore ignored.}}
\For{each $\tok$ in $M$}
	\If{$\tok.\textsc{rwCounter} < \rwLength$}
        \State Increment the \textsc{rwCounter} of each token by 1.
		\State Pick a neighbor $v$ uniformly and independently at random.
		\State Enqueue  $\tok$ into $\outbox_v$. 
	\Else
		\State Save node ID in {\texttt{verifiedList}}

        \State Send the verified token back to source of $\tok$ via the same path (in reverse).
        \Statex \Comment{\textit{Takes additional $\rwLength$ steps.}}
	\EndIf
\EndFor
\EndFor
\EndFor
\end{algorithmic}
\end{algorithm}

\subsubsection{Detailed Description} 
The Byzantine random walk protocol is designed so that most honest nodes can sample honest nodes almost uniformly at random in the presence of Byzantine nodes. Tokens sent by honest nodes perform random walks across the network, with each walk constrained by limitations on both the number of tokens allowed per edge and the maximum walk length, enforced through a time-to-live (TTL). At the start of each round, node $u$ initiates $\tot = \log^3n$ tokens. Each token, initially assigned a zeroed $\textsc{rwCounter}$, starts at node $u$ and includes a unique ID and the source node's ID.

The protocol restricts the number of tokens transmitted through each edge per round to $\capa = a \log^3 n $, with $a$ chosen large enough to control token flow. This is crucial to control the behavior of the Byzantine nodes, which can send an arbitrary number of tokens (and need not follow random walk protocol). Each honest node in the network is limited to sending and receiving up to $\capa$ tokens per round.  Each token can traverse a maximum distance of $ \rwLength = c \log n$ steps. The $\textsc{rwCounter}$ records the distance a token has traveled, incrementing after each step, and if it reaches $\rwLength$, the token stops and is marked as verified. 

 For $ c \log n $ steps, node $ u $ dequeues up to $\capa$ tokens from its $\outbox_u$ and distributes them by selecting a random neighbor $ v $ for each one and placing it in $ v $’s FIFO buffer, $\outbox_v$.  It may also receive up to $\capa$ tokens from its neighbors, storing them in a set $ M $. 
If any neighbor sends more than $\capa$ tokens, that neighbor is {\em blacklisted}, and further tokens from it are ignored. Each received token in $ M $ is processed: if its $\textsc{rwCounter}$ is less than $\rwLength$, it continues the random walk by selecting a new random neighbor. If the $\textsc{rwCounter}$ reaches $\rwLength$,  the token is marked as {\em verified}, and the token ID and current node ID are saved in the destination node. The verified token is then returned to its source (origin) node by retracing its path in reverse order (this will take an additional $\rwLength$ steps) --- Figure \ref{fig:nodet'}.   The source node uses its verified tokens to then subsequently make direct connections with the (respective destination) nodes that verified them. This verification crucially ensures that only nodes that sent tokens via random walks are eligible to make connections. This prevents Byzantine nodes (who know the addresses of all nodes in the full information model) from making arbitrary connection requests to honest nodes. Thus, even though Byzantine nodes can flood honest nodes with many connections, only connections made by nodes with a verified token will be entertained (note that the destination nodes record the source node's ID as well as its token ID).

\section{Protocol Analysis}

In this section, we analyze the P2P construction protocol and show that it builds and maintains
a large bounded-degree expander network among the honest nodes. The protocol crucially uses the Byzantine Random Walk Protocol to sample nodes nearly uniformly at random, which in turn is used to make new connections.

\subsection{Analysis of Byzantine Random Walk Protocol on Dynamic Graphs}

In this section, we make the assumption that at  time $t = \Theta(\sqrt{n})$, with high probability, there is a large subgraph (of size $|V_t| - o(|V_t|)$) of $G_t$ consisting of only honest nodes that is constant-degree {\em expander} graph. We show in the next section that our P2P protocol maintains this invariant over time. We will later prove that the above assumption  holds when  $t  =\Theta(\sqrt{n})$.

We note that when $t = \Omega(n)$, the network size would have stabilized, i.e., $|V_t| = \Theta(n)$ with high probability (in $n$) (cf. Lemma \ref{lem:sizeofg}). Hence, in the long run,  we will show that there is a large expander subgraph of honest nodes of size $n-o(n)$ in $G_t$, with probability at least $1-1/n$.


For analysis, we define a subgraph over the set of dynamic graphs called the {\em core} corresponding to a phase of our Byzantine random walk protocol. A phase refers to a single iteration of our protocol as defined below.  

\begin{figure}
    \centering
    \includegraphics[trim={0 0.25cm 0 0.5cm},clip, width=0.75\linewidth]{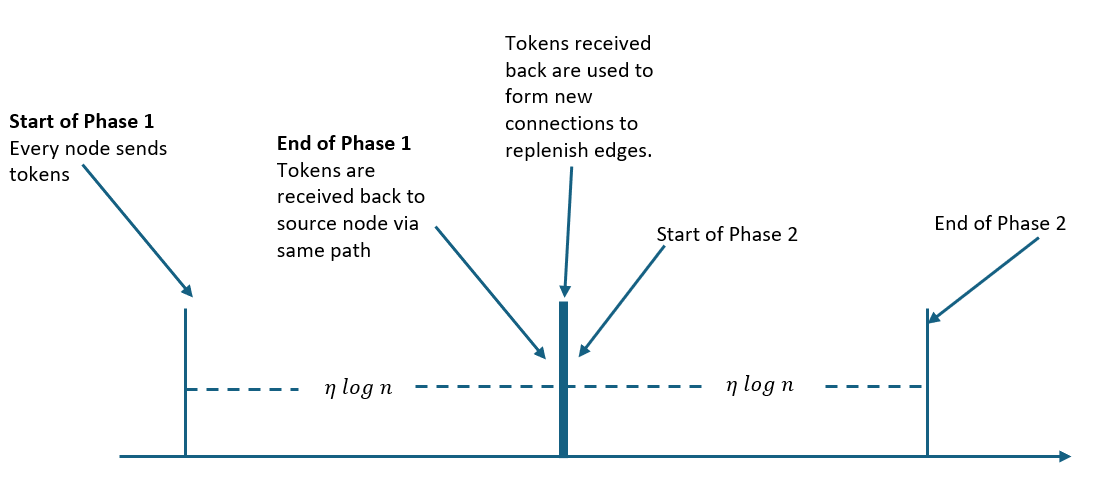}
    \caption{The figure illustrates how the tokens are being sent in the starting of a phase and reach a random node via random walk and sent back to the source node via the same path. The total time taken is $\eta \log n > 2\rwLength $ steps, for a suitably large constant $\eta$ given by the Byzantine Random Walk Theorem (item 3 in Theorem \ref{thm:ByzSamplingFinal}). At the end of a phase, these tokens are used to form new connections to replenish edges of a node where the formation of these new edges is instantaneous.}
    \label{fig:nodet'}
\end{figure}

\begin{definition}[Phase, denoted by $P_{t',t}$]\label{def:phase}
    Consider the Byzantine Random Walk Protocol. We define a phase as one iteration of the for loop (at line 1 from Algorithm \ref{alg:byzantineSamplingSparse}) in the protocol, such that the iteration begins at a integer time $t'$ and ends at time $t > t'$, where $t-t' = \eta \log n \geq 2\texttt{rwLength}$ steps (see Figure \ref{fig:nodet'}). This is the time interval that it takes for a random walk initiated by a node at time $t'$ to reach its destination $\texttt{rwLength}$ steps and return back (if successful). 
\end{definition}

We assume that there exists a large expander subgraph of honest nodes in $G_t$, at any time $t' < t$.
(We will later show that our P2P protocol maintains this invariant over time inductively.)

\begin{definition}[Core Subgraph of phase $P_{t',t}$, denoted by $\core$.] \label{def: core}
    Fix a phase $P_{t',t}$ (see Figure \ref{fig:core}). We assume that there exists a core subgraph $\core = (V_{\core}, E_{\core})$ is defined as the {\em expander}  subgraph of {\em honest} nodes of $G_{t'} - B_{t'} - D_{t',t}$ of size at least $|V_t'|-O(|B_{t'}|) - O(|D_{t',t}|)$, where $B_{t'}$ is the set of Byzantine nodes
    and $D_{t',t}$ is the set of nodes (including, possibly, Byzantine nodes) that joined and left the network in the time interval $[t',t]$.
    Since $t-t' = \Theta(\log n)$, it follows that $D_{t',t} = \Theta(\log n)$ (cf. Lemma \ref{lem:arrival}),
    and $B_{t'} = o(|V_t'|/\log(|V_t'|))$. Hence, the core subgraph is of size at least $|V_t'|- o(|V_t|/\log(|V_t'|)) - O(\log n) = |V_t| - o(|V_t|)$.
    
\end{definition}

\begin{figure}
    \centering
    \includegraphics[trim={0 0 0.1cm 0},clip, width=0.99\linewidth]{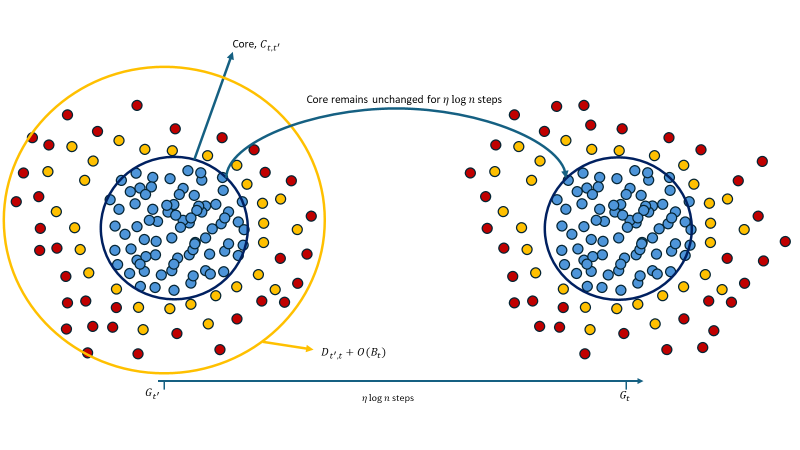}
    \caption{The figure illustrates the core subgraph in phase $[t', t]$ where $t-t' = \eta \log n$. The core nodes are shown in blue, the Byzantine nodes are shown in red, and the honest boundary nodes are shown in yellow. The core subgraph $C_{t,t'} = V_t - D_{t',t} - O(B_t')$, where $D_{t',t}$ is the set of all nodes that left or joined in the interval $[t',t]$ and $B_t'$ is the set of Byzantine nodes at time $t'$.  Byzantine Random Walk Theorem shows that most random walk tokens started from nodes in the core in $G_{t'}$ ends up in a near-uniform random node in the core, and the verified token returns to the source node in $G_t$. Further, the core is static during an entire phase.}
    \label{fig:core}
\end{figure}

Note from the definition \ref{def: core}, the core set $\core$ will be fixed throughout phase $P_{t',t}$. Thus, we reduce the dynamic graph to have a static core subgraph over which we can analyze the random walks. Hence, we can appeal to Theorem 2.1 of \cite{soda25}  to obtain the following result. This theorem says that most random walks started by honest nodes in the core subgraph will walk only in the core subgraph and will mix nearly uniformly at random in the core.
Specifically, most of the $\tot$ tokens created by most of the honest nodes in core will walk
$2\texttt{rwLength}$ steps in the core (time taken to reach its destination in $\texttt{rwLength}$
steps and then walk back in reverse after verification). Furthermore, the number of steps taken to complete 
these $2\texttt{rwLength}$ steps will be $O(\log n)$ with high probability. Further, each of these tokens has a nearly uniform probability
of $\Theta(1/|V_t|)$ of reaching any particular node in the core.

\begin{theorem}[Byzantine Random Walk Theorem --- adapted from \cite{soda25}]\label{thm:ByzSamplingFinal}
Let $\core$ be the core subgraph of $G_{t'}$ as defined in Definition \ref{def: core}, with a mixing time of $\tau_{C_{t',t}} = b\log n$. Suppose each (honest) node in $C_{t',t}$  initiate at most $\tot$ tokens and sends them in batches of $\capa = \Theta(\log^{3} n)$ via independent random walks for $\rwLength = 2\tau_{C_{t',t}}$ rounds using Algorithm \ref{alg:byzantineSamplingSparse}. Let $R(\core)$ denote the total number of tokens initiated by the nodes in the core $\core$. Then, for $\kappa = ((|B_t| + |D_{t',t}|) \log n)/|\core|$, the following statements hold whp:
\begin{enumerate}
    \item At most $O(\kappa \cdot  \tot \, |\core|)$ tokens enter or leave the core.

    \item At least $R(\core) - O(\kappa \cdot  \tot \, |\core|)$ tokens walk only in the core $\core$. Moreover, all of these tokens walk at least $\tau_C$ steps whp, and they finish their walks in at most $O\left(\frac{ \tot}{\capa}\tau_C\right) = \eta \log n$ rounds, where $\eta$ is a large enough constant.

    \item Furthermore, the probability that each such token ends at any given node $u \in \core$ is 
    $$\frac{deg_{\core}(u)}{2|E_{\core}|} \pm \frac{1}{|V_t'|^3} = \Theta(1/|\core|) = \Theta(1/|V_t'|),$$
    where $deg_{\core}(u)$ is the degree of node $u$ restricted to the core subgraph $\core$ and $E_{\core}$ is the edge set of $\core$.
    \item Finally, the verified token corresponding to each such token that ended at some $u \in C_{t',t}$ is successfully received back at $v$.
\end{enumerate}
\end{theorem}

Note that in the Byzantine Random Walk Theorem, we established that at most 
$O(\kappa \cdot \tot \cdot |C_{t',t}|$) tokens either enter or exit the core subgraph $C_{t',t}$, where $\kappa$ serves as an upper bound on the fraction of such tokens among the total number of tokens initiated by nodes in the core subgraph.

\subsection{Analysis of Node Joining in P2P Protocol}

We first prove a lemma that shows that a random sample from the \textit{nodesList} (maintained by the Entry Manager) is near-uniform 
random sample from the (current) network.

\begin{lemma}
\label{lem:entry-sample}
    Let $x \in V_t$ be a node (alive) in the network at any time $t = \Omega(\log n)$ and $r$ be a random sample from the \textit{nodesList}. Then, $\Pr(r = x) = \Theta\left( \frac{1}{|V_t|} \right)$.
\end{lemma}

\begin{proof}
    Fix a time $t = \Omega(\log n)$ and $ x \in V_t$. Let $\mathcal{L}_t$ denote the \textit{nodesList} maintained by the Entry Manager at time $t$. Recall from Algorithm \ref{alg:entrymgr} that the \textit{nodesList} maintains up to $n$ nodes at any given time, where $n$ is the stable network size. Moreover, when a new node joins, it is added to $\mathcal{L}_t$ without removing any existing nodes in the \textit{nodesList} if $|nodesList| < n$, else a node is removed uniformly at random from the current list. 
    To analyze the sampling probability of the node $x$, we consider the following two cases:
    \smallskip
    
    \noindent \textbf{Case 1: $t = o(n)$}. By Lemma \ref{lem:arrival}, the total number
    of nodes that arrived till $t$ is at most $n$ w.h.p., and hence the \textit{nodesList} is never full, resulting in every node arriving by time $t$ getting added to \textit{nodesList}. Furthermore, by Lemma \ref{lem:arrival}, the number of arrivals by time $t$ is $\Theta(t)$, w.h.p, for $t = \Omega(\log n)$. Also, by Lemma \ref{lem:sizeofg}, the number of nodes alive till time $\Theta(t) = |V_t|$, w.h.p.  Since $r$ is sampled randomly from the \textit{nodesList},  $Pr(r = x) = \frac{1}{\Theta(t)} = \Theta\left(\frac{1}{|V_t|}\right)$. 
    \smallskip
    
    \noindent\textbf{Case 2: $t = \Omega(n)$}. 
    By time $t$, the network has seen at least $\Theta(n)$ arrivals in expectation and w.h.p. Thus, \textit{nodesList} has $\Theta(n)$ entries.
    Also, by Lemma \ref{lem:sizeofg}, $|V_t| = \Theta(n)$.

    Thus, if $x$ is in the \textit{nodesList} at time $t$, it will be sampled
    with probability $(1/\Theta(n)) = \Theta(1/|V_t|)$.

    We show that with at least constant probability that $x$ will be in the 
    \textit{nodesList} at time $t$.

    Since $x \in V_t$, $x$ is still alive, and by Lemma \ref{lem:expo}, it 
    joined the network not more than $n$ time units ago with constant probability.
    During this time interval of $n$, by Lemma \ref{lem:arrival}, w.h.p. at most
    $\Theta(n)$ nodes arrived. We bound the probability that any of these nodes
    that arrived before $x$ could have replaced $x$ in the \textit{nodesList}.
    Note that an incoming node will sample a random node in \textit{nodesList} 
    and replace it.
    
     At any given time after $x$ is added to the \textit{nodesList}, the probability that $x$ is not replaced by a new (incoming) node in the \textit{nodesList} is $ 1 - \Theta(1/n)$. Since
     at most $\Theta(n)$ new nodes arrived since $x$ arrived,
    $Pr(x\in \mathcal{L}_t) = \left(1 - \Theta(1/n)\right)^{\Theta(n)}= \Theta(1).$

     Since $r$ is chosen uniformly at random from $\mathcal{L}_t$, 
        $$\Pr(r = x) = \Pr( x \in \mathcal{L}_t) \cdot \Pr( r = x | x \in \mathcal{L}_t) 
         = \Theta(1) \cdot \Theta(1/n) = \Theta\left( \frac{1}{|V_t|}\right).$$ 
    
    \end{proof}


The next lemma shows that most of the nodes in the \textit{nodesList} at time $t$
are alive at time $t$.

\begin{lemma}
\label{lem:alive}
Any node in the \textit{nodesList} at time $t$ has at least a constant probability of being
alive in $V_t$.
\end{lemma}

\begin{proof}
Consider a node $x$ in the \textit{nodesList} and assume that $x$ entered the network
at time $t$. Using a similar argument as in Lemma \ref{lem:entry-sample},
$x$ will be in the \textit{nodesList} for at most $\Theta(n)$ time with constant probability. Since the lifetime of a node is $\Theta(n)$ with constant probability,
$x$ will be in the current network for the duration it is in the \textit{nodesList}.
\end{proof}

Our next lemma is the main lemma of this section, which gives the probability  that a new incoming node succeeds in making at least $d$
honest connections to the network.

\begin{lemma}\label{lem:entry-manager}
    At any time $t$, for $d$ chosen to be a suitably large constant, the probability of a new node has least $d$ honest connections in $G_t$ is at least a constant in one connection attempt and is at least $1-1/n$ in $O(\log n)$ connection attempts.
    Furthermore, these honest connections are to nodes sampled almost uniformly at random from $V_t$.
\end{lemma}

\begin{proof}
    Fix a time step $t$. Let $u$ be the node that enters the network at time $t$, and $G_t$ be the network at time $t$. We note from steps 1-3 of Algorithm \ref{alg:p2pConstr} that a new node repeatedly attempts to connect to nodes $candidateList$ provided by the Entry Manager until it successfully connects to at least $d$ nodes in the network. We begin by analyzing the event that node $u$ connects to at least $d$ honest nodes in a single attempt. Recall that $u$ successfully connects to $d$ honest nodes in $G_t$ if the $candidateList$ consists of at least $d$ honest nodes that are part of $G_t$ and if at least $d$ of these honest nodes have incoming degree quota left to accept new incoming connections. 
    

    Each (honest) node in the network can have up to $3d$ outgoing connections and $6d$ incoming connections. Thus, the total number of incoming connections (among honest nodes) in the network is $d_{in} \leq 6d |V_t|$, and the total number of outgoing connections in the network is $d_{out} \leq 3d |V_t|$. Hence, at least $\frac 12$ fraction of nodes will have less than $6d$ incoming connections and thus will accept new incoming connection requests. So, when $u$ attempts to connect using $3d$ (outgoing) connections, the probability of selecting a {\em full} node (i.e., having $6d$ incoming degree) in any single connection attempt will be  $1/2$, assuming random selection. By Lemma \ref{lem:entry-sample},  the Entry Manager chooses $candidateList$ almost uniformly at random from the current network. Thus, each connection of $u$ is chosen almost uniformly at random from the current network. 
    Thus we note three facts --- (1) A node connects to a non-full node with constant probability;  (2) the fraction  of Byzantine nodes is $o(1)$; and (3) By Lemma \ref{lem:alive}, any candidate node is alive with constant probability. 
    
    Thus, the probability that a particular connection request of $u$ succeeds (i.e., connects
    to an alive, honest node that is not full) is at least a constant, say $\beta$.

    Thus, the probability that $u$ has at least $d$  connections to honest nodes in the network out of the possible $3d$ connections it tries is at least a constant.

     Therefore, the probability of node $u$ connecting to at least $d$ good nodes that are alive in a single connection attempt is at least $\Theta(1)$.  Hence, for $\Theta(\log n)$ connection attempts to the Entry Manager, we obtain a high success probability.

\end{proof}

\subsection{Analysis of Conductance}

In this section, we analyze and show the main guarantee of the P2P construction protocol,
namely, it maintains a large expander subgraph among honest nodes at all times.
We show this by induction on the number of {\em phases}.  We recall that the protocol (Algorithm \ref{alg:p2pConstr})
at the end of each {\em phase}, i.e., every $\eta\log n$ steps (where $\eta$ is a fixed constant, and $n$ is the stable network size),  drops old connections and makes new connections. 
We show that this keeps the conductance above a specific constant at all times (except possibly for a short initial phase), with high probability.  
The key lemma is given below. The lemma assumes that some time before stability is reached, say, $t_1 = \Theta(\sqrt{n})$, $G_{t_1}$ is an expander with high probability and then shows by induction that this expansion is maintained with high probability henceforth.\footnote{In the beginning of the stochastic churn process, say, when $t < \sqrt{n}$, we don't make any high probability claims on  $G_t$. This is not needed to show the desired result on $G_t$ in the long run.}  Subsequently, we will prove this assumption on $G_{t_1}$ and show that it is indeed an expander.

We need the following technical lemma for our conductance analysis (a proof is in the appendix).

\begin{lemma}[\cite{itcs24}]
\label{lem:probaInequality}
For any constant, $\Phi \in (0,1/10]$ and any two integers $n,s \geq 1$ such that $ s \leq n/2$, $\binom{n}{s} \binom{cs}{(1-\Phi) c s} (as/n)^{(1-\Phi) c s} \leq 1/n^2$ for large enough $n$ and some suitably chosen integer $c \geq 1$ and constant $a \geq 1$.
\end{lemma}

\begin{lemma}\label{lem:conductancelemma} 
    Fix a time $t_1  = k\eta\log n = \Theta(\sqrt{n})$, for a phase number $k$, and assume that there exists a subgraph of honest nodes in $G_{t_1}$ with size at least $(1-o(1))|V_{t_1}|$ and conductance at least $\phi'$, where $\phi' >0$ is a fixed constant. 
    Then for all $t > t_1$, with probability at least $1 -1/|V_t|$, there exists a subgraph of honest nodes in $G_t$ with size at least $(1-o(1))|V_t|$ and conductance at least $\phi$, where $0 < \phi <\phi'$ is a fixed constant.
    In particular, after network stability is reached, say, for any $t \geq 3n > t_1$, there exists a subgraph of
    honest nodes in $G_t$ with size $(1-o(1))n$ and conductance at least $\phi'$.
\end{lemma}
    
\begin{proof}
     The proof is by induction on $t$. The base case is $t_1 = kc\log n$ which is true as per the premise of the Lemma. We assume that the lemma is true at the beginning of phase $k'$, $t = k'\eta\log n > t_1$ 
       and show that it is true at the beginning of phase $k'+1$,  $t' = (k'+1)\eta\log n = t + \eta\log n$. By a property of expanders, it will follow that the premise will also hold for times in the (open) interval $(k'\eta\log n,(k'+1)\eta\log n)$, i.e., {\em between} the time steps of this phase.

     By the induction assumption, there is a subgraph $H_t$ of $G_t$ consisting of honest nodes at the beginning of phase $k'$ (where $t = k'\eta\log n$)  with size at least $(1-o(1))|V_t| = |V_t| - f|B_t|$ (for a suitably large fixed constants $f$ and $|B_t| = o(|V_t|/\log |V_t|)$) and conductance at least $\phi'$, where $\phi' > 0$ is a fixed constant.

     We consider the {\em core} subgraph  $C_{t,t+\eta\log{n}}$  (determined by the phase $[t,t+\eta\log n]$) of $G_t$ from Definition \ref{def: core},
     consisting of only good nodes in the network that joined before time $t$ and will remain in the network till time $t+\eta\log{n}$. Such a core subgraph $C_{t,t+\eta\log{n}}$ (referred henceforth to as just ``core'' for the rest of the proof) exists in $G_t$ and can be taken to be  $H_t - D_{t,t+\eta\log n}$,
     where $D_{t,t+\eta\log n}$ is the set of nodes that joined or left the network in the time interval $[t,t+\eta\log n]$.  We note that by Lemma \ref{lem:arrival}, $|D_{t,t+\eta\log n}| = O(\log n)$, and hence
     the size of core  is $|H_t| - O(\log n) = |V_t| - f|B_t| - O(\log n)$, by our assumption on size of $H_t$ (in previous para). The existence of the core subgraph is due to the property 
     of expanders (cf. Section \ref{subsec:defCore}) and it will have conductance at least $0 < \phi \leq \phi'$ for (any) constant $\phi$. Thus, the existence of this core subgraph  $C_{t,t+\eta\log{n}}$
     proves the theorem for time steps {\em between} this phase. We next show that at the
     beginning of the next phase, i.e., at phase $k'+1$ and corresponding time $t' = (k'+1)\eta\log n = t + \eta\log n$, $G_t'$ maintains the invariant, i.e., there is exists a subgraph of honest nodes in $G_{t'}$ with size at least $|V_{t'}| - f |B_t'|$ and conductance at least $\phi'$. Thus, at the end of
     each phase, we maintain the size and conductance invariant of the expander subgraph which will imply the lemma.

     To prove the above, we consider $C_{t,t+\eta\log{n}}$. At the end of phase $k'$ (i.e., at time
     $t+\eta \log n$), the protocol drops some connections and makes new connections.


    
    As stated above, the size of core $C_{t,t+\eta\log{n}}$ is $|V_t| - f|B_t| - O(\log n)$ (note
    that $|B_t|=o(|V_t|/\log(|V_t|)$.)
    Consider any subset $S \subset C_{t,t+\eta\log{n}}$ of size $q$ for $q \in \left[1, \frac{|C_{t,t+\eta \log{n}}|}{2} \right]$ and thus $q \leq |V_t|/2$. Note from step 2 to 6 of Algorithm \ref{alg:p2pConstr}, after a node (enters and) connects to the network, the number of outgoing connections are in $[d,3d]$ edges. Hence, the number of outgoing edges from the set of nodes in $S$ is in the range $[dq, 3dq]$. By the Byzantine Random Walk Theorem (cf. Theorem \ref{thm:ByzSamplingFinal}, item (3)), each outgoing edge has a $\Theta(1/|V_t|)$
    probability of landing in any particular node in the core.
    
    Let $E_S$ be the set of outgoing edges connected by nodes in set $S$. Label these edges as $e_1,e_2, \ldots, e_{|E_S|}$. Note that there are two ways the edges in $E_S$ might be present in the network: a) a node in $S$ connects to a node in $S$, or b) a node in $S$ connects to a node outside the set $S$. Using this observation, we have the following: 
    $$|E_S| = |E(S,S)| + |E(S,\overline{S})|.$$
    where $E(U,V)$ denotes the set of edges with a node in set $U$ connected by an outgoing edge to a node in set $V$. We can rewrite the above equation as 
    $$|E(S,\overline{S})| = |E_S| - |E(S,S)|$$
    
    This reduces our aim to compute an upper bound on the number of outgoing connections  of nodes in $S$ connected to some node in $S$. Thus, using the definition of conductance, we are required to show that 
    $$|E_S| - |E(S,S)| \geq \phi' |E_S| \quad \text{implies} \quad (1-\phi')|E_S| \geq |E(S,S)|$$
    where $\phi' > 0$ is the desired conductance.
    
    We define a set of indicator random variables for these events, i.e., for each edge $e_i$, defining the indicator random variable $X_i$ as 
    $$
        X_i = \begin{cases}
            1 , \, \text{ if edge }e_i \text{ has both endpoints in }S \\
            0, \, \text{ otherwise}
        \end{cases}
    $$

    Then, $X = \sum_{i = 1}^{|E_S|} X_i$ is the number of edges in set $S$ connecting to another node inside $S$.  
    As mentioned earlier, each edge from a node $ v \in S$ has probability at most $O(|S|/|V_t|)$ to land in $S$ or 
    $\Pr(X_1 = 1) \leq  O(|S|/|V_t|) = aq/n_t$, for some constant $a$ and assuming $|V_t| = n_t$.

    Next, we condition on the event that $X_1 = 1, \ldots , X_{i-1} = 1$, resulting in a decrease in probability for $X_i$ having both end points in $S$, then for all $ i = 2, \ldots , |E|$
    $$ \Pr\left( X_i = 1 \bigg| \bigwedge_{j=1}^{i-1} X_{j} = 1\right) \leq aq/n_t$$

    \noindent Applying chain rule of conditional probability, we have
    \begin{align}\label{eq:chainRule}
        \Pr\left(\bigwedge_{j=1}^{i} X_j = 1\right) &= \Pr\left(X_i = 1 \bigg| \bigwedge_{j=1}^{i-1} X_j = 1\right)\cdot \Pr\left(\bigwedge_{j=1}^{i-1} X_j = 1\right)\nonumber\\ 
        &= \Pr(X_j = 1) \cdot \prod_{j=2}^i \Pr\left( X_j = 1 \bigg| \bigwedge_{k=1}^{j-1}X_k = 1\right)\leq \left( a/n_t\right)^i
    \end{align}

    \noindent Next, we bound the probability of $|E(S,
    S)| \geq (1 -\phi')|E_S|$.  To do so, we consider all subsets $S$ of size $q$ (of size at most $n_t/2$) in the core set $C_{t,t+\eta\log{n}}$.  Then, we have 
    \begin{equation*}
        \Pr\left(\exists S, |S| = q \text{ and } |E(S,S)| \geq (1 - \phi')|E_S|  \right) \leq  {{n_t}\choose{q}} {|E_S| \choose {|E_S|(1 - \phi')}} \left( aq/n_t\right)^{(1 - \phi')|E_S|}
    \end{equation*}
    
    Since $|E_S| \leq  \delta q$, for some constant $\delta \geq 1$, we can appeal to Lemma \ref{lem:probaInequality} to bound the above probability to be $1/(n_t)^2$.

     Hence, by a union bound, over all set sizes $q \in [1, n_t/2]$, with probability at least $1-1/n^t$, the graph $G_t$ has a subgraph with conductance at least $\phi'$.

    Now we show that the size of the expander subgraph is as desired.  This is because the number of nodes that joined the network during this phase is $\Theta(\log n)$, which at least compensates
    for the reduced size of the core due to the leaving nodes that are excluded from the core. The new nodes connect to
    at least $d$  nodes supplied by the Entry Manager. From Lemma \ref{lem:entry-manager}, each such connection has a constant probability of connecting to the core. It is easy to show that any subset of the new nodes has conductance at least $\phi'$ for a suitably chosen $\phi'$. Since  the core subset has conductance at least $\phi'$, it follows that the core plus the set of
    $\Theta(\log n)$ new nodes has size at least $|V_t| - f|B_t|$ (for a suitable $f$) and has conductance at least $\phi'$.
     
\end{proof}

Next we show that the assumption made in Lemma \ref{lem:conductancelemma} can be satisfied. 

\begin{lemma}\label{lem:sqrtn}
    At time $t_1  = \Theta(\sqrt{n})$,  there exists a subgraph of honest nodes in $G_{t_1}$ with size at least $(1-o(1))|V_{t_1}|$ and conductance at least $\phi'$, where $\phi' >0$ is a fixed constant. 
\end{lemma}

\begin{proof}
    This lemma can be shown by making use of two facts: (1) By Lemma \ref{lem:expo}, the number of nodes leaving the network is negligible with high probability during the time interval $[0,t_1]$, since $t_1 = o(n)$. In particular, by Lemma \ref{lem:expo}, the probability that any node that joined the network during this time interval left the network is $O(1/\sqrt{n})$. Thus, with high probability, except for at most $O(\log n)$ nodes, all nodes are alive. (2) By Lemma \ref{lem:entry-manager},  each new node
    connects to at least $d$ (alive)  honest nodes in the network (note that $d$ is a sufficiently large constant) chosen almost uniformly at random.
    
    We consider the time interval $[\sqrt{n}, 2\sqrt{n}]$. It follows
    from Lemma \ref{lem:sizeofg}, with high probability, there are $\Theta(\sqrt{n})$ nodes alive at time $\sqrt{n}$ (call this set $S_1$) and $\Theta(\sqrt{n})$ nodes joined the network during this interval (call this set $S_2$).  Since nodes in $S_2$ connect to nodes in $S_1$ almost uniformly at random, it is easy to show (similar to the argument in Lemma \ref{lem:conductancelemma}) that all subsets of $S_2$ have conductance at least $\phi'$ for a suitable constant $\phi'$. 
    By symmetry, for $d$ sufficiently large, we can show that all subsets in $S_1$ have conductance at least $\phi'$.

\end{proof}

We can state our main result which follows from the above lemmas. 

\begin{theorem}
\label{thm:main}
    Assume that the churn is controlled by a stochastic churn model of Definition \ref{def:poisson_node_churn}. Let $n$ be the stable network size.
    For any $t \geq \Theta(\sqrt{n})$, the P2P Construction Protocol, with  probability at least $1-1/n^{\Omega(1)}$ maintains  a graph $G_t = (V_t,E_t)$
    that has a constant-degree expander subgraph of size at least $(1-o(1))|V_t|$ while tolerating up to $o(|V_t|/\log(|V_t|))$ Byzantine nodes.
    In particular, after stability is reached, i.e., $t\geq 3n$, with probability at least $1-1/n^{\Omega(1)}$, maintains  a graph $G_t = (V_t,E_t)$
    that has a constant-degree expander  subgraph of size at least $(1-o(1))n$ while tolerating up to $o(n/\log n)$ Byzantine nodes.
\end{theorem}

\section{Conclusion}
In this paper, we have presented a randomized, fully distributed P2P construction protocol that,  with high probability, constructs and maintains a bounded-degree graph with high expansion in the presence of continuous churn and Byzantine nodes.  Our protocol can be used to enable a variety of distributed tasks in dynamic networks with Byzantine nodes, such as efficient routing, broadcasting, agreement and leader election.  

Several questions remain open. First, while our protocol can tolerate up to $o(n/ \polylog n)$ Byzantine nodes, can we further increase the number of Byzantine nodes that can be tolerated? Second, can we extend our approach to an adversarial churn model (e.g.,\cite{podc13}) rather than the stochastic churn model assumed here?


\paragraph{Acknowledgements.}
Gopal Pandurangan was supported in part by ARO grant W911NF-231-0191 and NSF grant CCF-2402837.

\bibliography{refs}

\newpage

\section{Appendix}

\subsection{Proof of Lemma \ref{lem:arrival}}
\label{sec:arrivalproof}

\begin{proof}
    In a Poisson process with rate $\lambda = 1$, the number of nodes arriving in any interval of length $t' - t$ is Poisson distributed with mean 
    $ \lambda (t - t') = 1 \cdot (t - t') = t - t'$. Hence, $E[N(t',t)] = (t' - t)  = n'$.

To prove concentration of $N(t',t)$ we use a Chernoff bound for the Poisson distribution \cite{upfal}. 
    
    The Chernoff bound for the Poisson distribution gives us:
        \begin{align*}
            Pr(|N(t',t) - n'| \geq \delta n') &\leq 2e^{-n' \delta^2 / 3}
        \end{align*}
        for any $0 < \delta < 1$.

        Plugging  $\delta = \sqrt{\frac{3 \ln(2n^2)}{n'}}$, we have the lemma.
\end{proof}

\subsection{Proof of Lemma \ref{lem:probaInequality}}

\begin{proof}
    Let $p^* = \binom{n}{s} \binom{cs}{(1-\Phi) c s} ((1+\varepsilon)s/n)^{(1-\Phi) c s}$. We give two upper bounds: the first for $s = o(n)$ and the second for $s = \kappa n$ for some constant $0 < \kappa \leq 1/2$.
    
    To get the first bound, we use the inequality $\binom{y}{x} \leq (ey/x)^x$, that holds for any integers $y \geq x \geq 1$. Then, $p^* \leq (en/s)^s (e/(1-\Phi))^{(1-\Phi) c s} ((1+\varepsilon) s/n)^{(1-\Phi) c s} = 2^{s(\beta + (1-(1-\Phi)c) (\log n - \log s))}$, where $\beta = \log(e) + (1-\Phi)c \log (e (1+\varepsilon) / (1-\Phi)) $. For large enough $n$, the $(1-(1-\Phi) c) \log n$ factor dominates in the exponent. Thus, 
    it suffices to choose $c$ large enough and it holds that $p^* \leq 1/n^2$ for large enough $n$.
    
    To get the second bound, we need to use a tighter inequality (since $s$ and $n$ are large) to bound the binomial coefficients. More concretely, using Stirling's formula, it holds that $\binom{y}{x} \leq 2^{y H(x/y)}$, for any integers $x, y \geq 0$ and where $H(q) = -q \log(q) - (1-q) \log (1-q)$ is the binary entropy of $q \in (0,1)$. As a result, we get $p^* \leq 2^{n ( H(\kappa) + c \kappa H(1-\Phi) + (1-\Phi) c \kappa \log ((1+\varepsilon) \kappa))}$. 
    First, 
    we take small enough $\varepsilon$ such that for any $\kappa \leq 1/2$, $|\log ((1+\varepsilon) \kappa)| \geq 0.9$. 
    Next, note that 
    $H(1-\Phi) \leq 2 \sqrt{\Phi (1-\Phi)} \leq 0.6$, where the first inequality is a well-known upper bound for binary entropy that holds for any $\Phi \in (0,1)$ and the second one holds because $2 \sqrt{x (1-x)}$ takes value $0.6$ at $x=1/10$ and increases between $x=0$ and $x=1/2$. Then, $H(1-\Phi) \leq |(3/4) (1-\Phi) \log ((1+\varepsilon) \kappa)|$, since the right-hand side is strictly greater than $0.6$ for $\Phi \leq 1/10$.
    Therefore, $p^* \leq 2^{n ( H(\kappa) + (c \kappa/4) (1-\Phi) \log ((1+\varepsilon) \kappa))}$. Next, we take $c$ large enough so that $|(c \kappa/8) (1-\Phi) \log ((1+\varepsilon) \kappa)| \geq 2 |\kappa \log \kappa|$. Then, $p^* \leq 2^{n ( \kappa \log \kappa - (1-\kappa) \log(1-\kappa) + (c \kappa/8) (1-\Phi) \log ((1+\varepsilon) \kappa))}$. Since $\kappa \log \kappa - (1-\kappa) \log(1-\kappa) \leq 0$ for $0 \leq \kappa \leq 1/2$ and the other term in the exponent is negative,  $p^* \leq 2^{n (c \kappa/8) (1-\Phi) \log ((1+\varepsilon) \kappa)}$. Finally, we use again $\log ((1+\varepsilon) \kappa) \leq -0.9$ to obtain $p^* \leq 2^{- b \cdot s }$, where $b = 0.9 c (1-\Phi) / 8$ does not depend on $\kappa$. Since $s = \Omega(n)$, it holds that $p^* \leq 1/n^2$ for large enough $n$. 
\end{proof}

\subsection{Proof of Theorem \ref{thm:ByzSamplingFinal}}

The proof of Theorem \ref{thm:ByzSamplingFinal} is an adaptation of the proof of the Byzantine Random Walk Theorem
of \cite{soda25} to the stochastic churn model with Byzantine nodes. For the sake of completeness, we give the full proof here.

\begin{lemma}[\cite{soda25}]\label{lem: coremixing}
    Fix a phase $P_{t',t}$. Let $\core $ be the core subgraph during this phase. Then, all random walks initiated at the start of $P_{t',t}$, and that walk only on $\core$, will walk at least $c \log n$ steps (whp) and hence will mix in $\core$, where $c$ is a large enough constant (depending on $\phi_{\core}$).
\end{lemma}
\begin{proof}
    Fix a phase $ P_{t',t}$. Consider a random walk that walks only in $\core$, or more precisely, that walks on nodes $u_1, u_2, u_3, \ldots, u_i, \ldots, u_t$ during the phase $P_{t',t}$, with $u_1$ being the node that initiated the random walk and $u_t$ being the node where it terminated. When the walk enters each $u_i$, there are at most $d \cdot \capa$ random walks that are allowed to enter (from all incident edges) because $u_i$ will discard any excess walks and blacklist any neighbor having sent more than $\capa$ tokens. 

    Then, the random walk is placed on $\outbox_{u_{i+1}}$ that was chosen randomly by $u_i$. Of course, every other walk (regardless of whether it was initiated by an honest node or a Byzantine node) is also placed randomly in one of the $d$ outboxes. Therefore, even assuming the full set of $d\cdot \capa$ walks arrived at $u_i$ along with (and including) the random walk, the number walks placed into $\outbox_{u_{i+1}}$ is a binomial random variable with parameters $d.\capa$ and $1/d$, thus having a mean of $\capa$. Importantly, when $\capa$ is $a\log^3n$, the probability that the number of walks placed into $\outbox_{u_{i+1}}$ will exceed $a\log^3n +(a/2b) \log^2 n = (1 + 1 /(2b \log n)) \cdot a\log^3 n$ is at most
    $e^{(a \log^3 n) \cdot (1/(4b^2 \log^2 n))/3} = n^{-1/(12b^2)}$
        by a  Chernoff bound. (Note that since $b \geq 1,1/(2b \log n) \leq 1$ for any $ n \geq 2)$. As a result, with probability at most $n^{-1/(12b^2)}$, the excess number of random walks placed into $\outbox_{u_{i+1}}$ (i.e., in addition to the mean $\capa$) is at most $(a/2b)\log^2 n$ per round. With $\rwLength = 2f = 2b \log n$, the number of such excess walks placed in $\outbox_{u_{i+1}}$ in the whole phase is at most $2b \log n \cdot (a /2b) \log^2 n = \capa$ with probability at most $n^{-a/(12b^2)}$, or in other words, whp for a constant $a$ chosen large enough compared to $b$. 

    With the excess smaller or equal to $\capa$ (whp), the random walk will buffer at $\outbox_{u_{i+1}}$ for at most one round before moving on to $u_{i+1}$. Thus, even if the random walk were unlucky and took two rounds at each node, in $2f$ steps, it would have taken the requisite $f$ random walk steps to ensure mixing. Thus, we can ensure that all random walks that only walked in $\core$ for $2f$ rounds will take at least $f$ random walk steps. 
\end{proof}
In the next lemma, we bound the number of edges in the core subgraph for a given phase. This hinges on the fact that due to the stochastic churn model, the graph does not change drastically over $\Theta(\log n)$ steps. 

\begin{lemma}\label{lem:sizeofcore}
    Fix a phase $P_{t',t}$, where $t' \geq 3n$ and nodes arrivals and departures are accounted for dynamically. Let $\core$ be the core subgraph for this phase. Then, $|{\core}|  = \Theta(n)$.
\end{lemma}
\begin{proof}
    Fix a phase $P_{t',t}$. We first bound the number of nodes in $G_{t',t}$. Let $V_{t'}$ be the vertex set for graph $G_{t'}$ at time $t'$. Then, from Lemma \ref{lem:sizeofg}, we know that the number of nodes in the network at time $t'$ is $\Theta(n)$ whp. Also,  $t-t' = \Theta(\log n)$ from definition \ref{def:phase}, then from Lemma \ref{lem:arrival}, whp at most $\Theta(\log n)$ nodes depart or join from the network during this phase. Thus, whp,
    $$G_{t',t} = \Theta(n) - \Theta(\log n) .$$ 
    Recall from Definition \ref{def: core} that $\core = G_{t',t} - O(B_t)$. Since $G_{t'}$ is an expander graph, it has strong connectivity properties and high conductance, $\phi$, ensuring that removing a small fraction of nodes does not break expansion whp. Then, using the fact that $|B| = o(n)$ as the number of Byzantine nodes are small relative to the network size, we have 
    \begin{equation*}
        |\core| = \Theta(n) - \Theta(\log n) - o(n) = \Theta(n). \qedhere
    \end{equation*}   
\end{proof}

In the next lemma, we bound the probability that a random walk starting at any given node in the $\core$ ends at another node within the $\core$ after $\Theta (\log n )$ steps.

\begin{lemma}[\cite{soda25}]\label{lem: mean}
    Let $\mu = |B_t|/|\core|$. Then, 
    
    $$\frac12 |\core| \leq (1 -O(\mu))d \leq |E_\core| \leq \frac 12 d |\core|$$

    where $E_{\core}$ is the set of edges in core subgraph $\core$ in phase $P_{t',t}$.
\end{lemma}
\begin{proof}
    Fix a phase $P_{t',t}$, where $t' \geq 3n $, and let $\core$ denote the core subgraph for this phase. From Definition \ref{def: core} and Lemma \ref{lem:sizeofcore}, we know that $|\core| = \Theta(n)$, and the core consists of all honest nodes that neither joined nor departed during the interval this phase and are not Byzantine. 

    Let $G_{t'}$ be the graph at time $t'$, and assume that it is a $d$-regular graph (as per Algorithm \ref{alg:p2pConstr}). Let $d_{max}$ be the maximum degree of any node, and $d_{avg}$ be the average degree. The algorithm ensures that the bounded degrees, so $d_{max} = O(1)$ and we assume $d_{avg} = \Theta(d_{max})$.

    First, we establish the upper bound for $|E_{\core}|$. Each node $v \in V(\core)$ has $\text{deg}_{C_{t',t}}(v) \leq d_{max}$. The sum of degrees in $\core$ is $2|E_{\core}|$. Thus, $2|E_{\core}| \leq |\core|d_{max}$. Using $d_{max}=O(d)$, we have:
    
    $$ |E_{\core}| \leq \frac{1}{2}d|\core| $$

    Next, we establish the lower bound for $|E_{\core}|$. Assuming the nodes of the $\core$ have $d_{avg}$ degree, then the number of edges in $\core$ would be $\frac 12 d_{avg} |\core|$. But some these edges might have been affected by Byzantine nodes or some of the nodes might have departed due to the churn during this phase. Therefore, number of such unavailable nodes is $|B_t| + |D_{t',t}|$. This means some of the node's edges would be outside the $\core$ and each such node could have been connected to $d_{max}$ nodes. Each such node could have been connected to at most $d_{max}$ nodes within $V(C_{t',t})$. Thus, the maximum number of edges ``removed'' or not formed with respect to $V(C_{t',t})$ is bounded by $d_{max}(|B_t| + |D_{t',t}|)$. 

    From Lemma \ref{lem:arrival}, $|D_{t',t}| = O(\log n)$. Substituting $d_{avg} = \Theta(d)$ and $d_{max} = O(d)$:
    $$|E_{C_{t',t}}| \geq \frac{1}{2}d|C_{t',t}| - O(d(|B_t| + \log n)).$$Factoring $\frac{1}{2}d|C_{t',t}|$:$$|E_{C_{t',t}}| \geq \frac{1}{2}d|C_{t',t}| \left(1 - O\left(\frac{|B_t|}{|C_{t',t}|} + \frac{\log n}{|C_{t',t}|}\right)\right).$$

    With $\mu = |B_t|/|C_{t',t}|$ and $|C_{t',t}| = \Theta(n)$, the term $\frac{\log n}{|C_{t',t}|} = O\left(\frac{\log n}{n}\right)$. Given that $|B_t|$ is $\omega(\log n)$ and  $O\left(\frac{\log n}{n}\right)$ is $O(\mu)$. Thus, we derive the lower bound:
    $$|E_{C_{t',t}}| \geq \frac{1}{2}d|C_{t',t}|(1 - O(\mu)).$$

    As lemma the states $(1 - O(\mu))d \le |E_{C_{t',t}}|$. Comparing this with the derived lower bound $|E_{C_{t',t}}| \ge \frac{1}{2}d|C_{t',t}|(1 - O(\mu))$, the lemma's inequality $(1 - O(\mu))d \le |E_{C_{t',t}}|$ holds if $(1 - O(\mu))d \le \frac{1}{2}d|C_{t',t}|(1 - O(\mu))$. Assuming $d > 0$ and $(1-O(\mu)) > 0$, this simplifies to $1 \le \frac{1}{2}|C_{t',t}|$, or $|C_{t',t}| \ge 2$, which is true for any non-trivial core.

    Combining the derived bounds on $|E_{C_{t',t}}|$ yields the full statement of the lemma.
    
\end{proof}

\begin{lemma}[\cite{soda25}]\label{lem: corewalk}
    Fix a phase $P_{t',t}$. Consider, for any core node $v \in \core$, any random walk that starts at $v$ and walks only through the nodes of $\core$ for $\Theta(\log n)$ steps. Then, the probability that the walk is at node $u \in \core$ at the end of the random walk is $\Theta \left(\frac{1}{|\core|} \right) = \Theta \left(\frac1n \right)$.
\end{lemma}
\begin{proof}
    Consider a random walk that starts at a node $v \in \core$ and walks only on nodes in $\core$. Conditioning that the walk only uses edges in $\core$, then it holds that any node in $\core$, the walk chose a \textit{uniform random outgoing edge among edges in } $\core$. Hence, the conditioned random walk is a standard random walk on $\core$. Since $\core$ is an expander, or more precisely $\core $ has constant conductance $\phi_{\core}$, the (conditioned) random walk on $\core$ mixes in $\tau_\core$ steps and reaches close to the \textit{stationary distribution} in $\core$ (up to $1/n^3$ error, as we have defined in Subsection \ref{subsec:defCore}). In particular, the stationary probability of node $u \in \core$ is $deg_{\core}(u/2|E_\core|)$. Now, we can apply Lemma \ref{lem: coremixing} where $\core$ is the honest subset of nodes and $f = \tau_\core = b\log n$ is the mixing time of $\core$. As a result, the probability that the walk ends up at a node $u \in \core$ is proportional to its degree $deg_\core(u)$ (where $1 \leq deg_\core(u) \leq d)$ divided by the number of edges in $\core$ (and up to $1/n^3$ error)  which is $\Theta(d|C|/2) = \Theta(n)$ (by Lemma \ref{lem: mean}, and as $|B_t| - o(n)$ and $|\core| = n - O(|B_t|) = \Theta(n)$. Hence, the probability that the walk ends at $u$ is $\Theta(1/n)$.

\end{proof}

\begin{lemma}[Verified Random Walk in $\core$]\label{lem: verifiedrw}
    Fix a phase $P_{t',t}$, where $t' \geq 3n$. Consider for any node $v$ in $\core$ if a random walk that started at $v$ and walked only through the nodes of $\core$ for $\Theta( \log n)$ time steps, and ended at some node $u$ in $\core$. Then, the probability that the verified token is successfully received back at $v$ successfully is at least $1 - o(1)$.
\end{lemma}

\begin{proof}
    Fix a phase $P_{t',t}$, and consider a random walk starting from node $v$ in $\core$. The walk proceeds through nodes $x_1, x_2, \ldots x_k$ where $k = \Theta (\log n)$, and stays entirely within $\core$. 

    Recall that $t' - t = \Theta (\log n)$ from Definition \ref{def:phase} and from Lemma \ref{lem:arrival}, the expected number of departures in the entire graph during this phase, $P_{t',t}$ is $\Theta(\log n)$.  The verified stage fails when even one of the oldest nodes is one of the nodes in the random walk. Fix a node $x_i$ on the random walk. Then, the probability that $x_i$ belongs to the set of oldest $\Theta(\log{n})$ nodes is 
    $$ \frac{\Theta (\log n)}{|\core|} = O\left( \frac{\log n}{n} \right)$$
    where the above follows from Lemma \ref{lem:sizeofcore} that $|\core| = \Theta(n)$. Since the random walk has length $k = \Theta(\log n)$, the probability at none of the nodes in $x_1,...,x_k$ depart is
    \begin{equation*}
        \left(1- \frac{\log n}{n} \right)^{\log n} = 1-O\left(\frac{\log^2n}{n}\right) \approx 1 - o\left(1 \right) \qedhere
    \end{equation*}
\end{proof}

\begin{lemma} \label{lem: kappa}
    Fix a phase $P_{t',t}$. Let $\kappa = (|B| \log n) / |\core|$, and $R(\core)$ denote the total number of tokens initiated by the nodes in the core $\core$ during this phase. Recall that each (honest) node initiates a maximum of $\capa$ (good) tokens in each phase of Algorithm \ref{alg:byzantineSamplingSparse}. Then, at most $O( \kappa |\core| \capa)$ tokens enter or leave $\core$, and at least $R(\core) - O(\kappa |\core| \capa)$ tokens walk only in $\core$ (i.e., are good). Moreover, these good tokens walk at least $\tau_\core$ steps whp.
\end{lemma}
\begin{proof}
    Fix a phase $P_{t',t}$. Recall from definition \ref{def:phase} that $t-t'= \Theta(\log{n})$. Thus, from Lemma \ref{lem:arrival}, the expected number of nodes that depart during this phase is $\Theta(\log n)$. Since the graph remains an expander throughout the phase, we consider the cut between the core $\core$ and the rest of the network. 

    From the definition of $\kappa$, we know that the fraction of Byzantine nodes in $C_{t',t}$ is at most $\kappa$. Also,  since we assume that at any time $t$, $G$ is a $d$-regular expander, thus the total number of edges between the core and the rest of the network is $O(|B|)$. So, the number of tokens that can cross this boundary in one step is $O(|B| \log^3 n)$, as each honest node forwards at most $\capa = O(\log^3 n)$ tokens. 

    Thus, the total number of tokens that enter or leave the core is $O(\kappa \cdot \core \cdot \capa )$ and the total number of tokens remaining inside is $R(C) - O(\kappa \cdot \core \cdot \capa).$ These tokens are unaffected by Byzantine nodes and complete their random walks unbiased. 

    As the core remains an expander throughout the phase, the mixing time $\tau_{\core}$, ensures the token movement follows the Algorithm \ref{alg:byzantineSamplingSparse} and they complete at least $\tau_\core$ steps before termination whp. Thus, whp, the tokens that remain in $\core$ mix properly.
\end{proof}

From Lemma \ref{lem: coremixing}, \ref{lem:sizeofcore}, \ref{lem: corewalk}, and \ref{lem: kappa} together imply the Theorem \ref{thm:ByzSamplingFinal}.

\end{document}